\documentclass[a4paper,12pt]{article}

\pdfoutput=1

\usepackage[tbtags]{amsmath}
\usepackage{amssymb,amsfonts,amsthm}
\usepackage{setspace}
\usepackage{graphicx}
\usepackage{natbib}
\usepackage{url}
\usepackage{centernot}

\newtheorem{theorem}{Theorem}

\DeclareMathOperator{\E}{E}
\DeclareMathOperator{\diag}{diag}
\DeclareMathOperator{\Var}{Var}
\DeclareMathOperator{\Cov}{Cov}
\DeclareMathOperator{\zeros}{zeros}

\newcommand{\given}{\mathbin{\vert}\nolinebreak}
\newcommand{\trans}{^{\scriptscriptstyle T}}
\newcommand{\mtrans}{^{\scriptscriptstyle -T}}
\newcommand{\unit}[1]{\ensuremath{\, \mathrm{#1}}}
\newcommand{\bs}[1]{\boldsymbol{#1}}

\newcommand{\bE}{\bs{E}}
\newcommand{\bX}{\bs{X}}
\newcommand{\bY}{\bs{Y}}
\newcommand{\bZ}{\bs{Z}}

\newcommand{\ldef}{\mathrel{:=}\nolinebreak}

\newcommand{\norm}[1]{\left\Vert #1 \right\Vert}

\newcommand{\geqquery}{\buildrel ? \over \geq}
\newcommand{\Perp}{\ensuremath{\mathrel{\perp\!\!\perp}}\nolinebreak}

\title{Computation and Visualisation for large-scale Gaussian updates}

\author{Jonathan Rougier\thanks{Reader in Statistics.  Address: School
    of Mathematics, University of Bristol, University Walk, Bristol
    BS8 1TW, UK. Email \texttt{j.c.rougier@bristol.ac.uk}.}\\School of
  Mathematics\\University of Bristol\\[1ex] \and Andrew Zammit Mangion
  and Nana Schoen\\School of Geographical Sciences\\University of
  Bristol}

\date{}

\begin{document}

{
\maketitle

\begin{abstract}
  \noindent%
  In geostatistics, and also in other applications in science and
  engineering, we are now performing updates on Gaussian process
  models with many thousands or even millions of components.  These
  large-scale inferences involve computational challenges, because the
  updating equations cannot be solved as written, owing to the size
  and cost of the matrix operations.  They also involve
  representational challenges, to account for judgements of
  heterogeneity concerning the underlying fields, and diverse sources
  of observations.

  Diagnostics are particularly valuable in this situation.  We present
  a diagnostic and visualisation tool for large-scale Gaussian
  updates, the `medal plot'.  This shows the updated uncertainty for
  each observation, and also summarises the sharing of information
  across observations, as a proxy for the sharing of information
  across the state vector.  It allows us to `sanity-check' the code
  implementing the update, but it can also reveal unexpected features
  in our modelling.  We discuss computational issues for large-scale
  updates, and we illustrate with an application to assess mass trends
  in the Antarctic Ice Sheet.
  
  \bigskip\noindent%
\end{abstract}

} 

\section{Introduction}
\label{sec:introduction}

Statisticians are now attempting inferences of a scale and complexity
that were unthinkable even a few years ago.  This is for a number of
reasons:
\begin{enumerate}
\item Computers are more powerful, and have larger memories,
\item New statistical techniques are available to represent judgements
  on large collections of random quantities, and to compute on those
  judgements,
\item Large new datasets, including from remote sensing, are becoming
  available, and
\item There is a political need, and research funding, to address
  inference for complex systems, notably environmental systems.
\end{enumerate}
Similar assessments have been given by \citet[ch.~1, concerning
meteorology]{kalnay02} and \citet[ch.~1, concerning decision
support]{smith10}.  In our application, outlined below, the state
vector has about $10^5$, and there are $3.5 \times 10^5$ observations.
Statistical inferences of this scale are most easily handled using a
Gaussian process prior, and the linearisation of the observation
operator; or else the use of an optimisation approach that comes to
very much the same thing \citep[e.g., as in data assimilation for
meteorology, see][]{apte08}. 

One concern in a complex inference is to verify that the code and the
model are performing sensibly, and a second is to visualise the
assimilation of very large numbers of observations.  The latter is
particularly challenging over multiple interacting processes.  We
present a visual diagnostic which meets both of these needs, based on
an upper bound on the updated variance.  It is almost obvious that the
updated variance of any measured linear combination of the state
vector has to be no larger than the smaller of its initial variance
and the observation error variance.  A scalar version of this result
was the basis for L.J.~Savage's principle of stable inference
\citep{savage62,edwards63}, in which an observation of $Y$ with error
variance $\tau^2$ had the effect of setting the updated variance of
$Y$ to $\tau^2$, regardless of the initial variance, provided only
that the initial distribution was suitably flat in the region around
the observation.  Savage used this as a demonstration of how
subjective Bayesian scientists might end up agreeing despite their
different initial judgements.  Below, we prove a similar result for a
collection of observations (section~\ref{sec:limit}).

The simplest use of our result is to trap errors, since violation of
the upper bound is \textit{prime facie} evidence of a failure in the
code.  Such failures are unlikely to occur for moderately-sized
problems, for which the matrix algebra can be programmed more-or-less
as written, but are an issue for large problems, where the calculation
must be broken into chunks, or approximated, e.g.\ using iterative
solvers.  But there is additional information available in the source
of the bound (initial variance or observation error variance), and in
the relation of the updated variance to its bound.  This leads
naturally to a visualisation tool in updates of random fields, for
which the linear combinations are often localised in the domain.  As
no reference is made to the value of the observations, this diagnostic
can be used before the observations are made available, for example in
experimental design \citep[e.g.][]{krause08}.

Section~\ref{sec:result} describes the theoretical result and its
implications, its implementation, and the `medal plot' for
visualisation.  Section~\ref{sec:computation} discusses computation
for large-scale applications, including a powerful result concerning
sparsity.  Section~\ref{sec:AIS} illustrates with an inference for
mass trends in the Antarctic Ice Sheet.

\section{Result and implications}
\label{sec:result}

\subsection{An upper bound for the updated variance}
\label{sec:bound}

Let $\bX$ be the collection of Gaussian random quantities, and $\bY
\ldef A \bX$ be the known linear combinations which are measured,
where $A$ is sometimes termed the `incidence matrix'.  Let $\bZ \ldef
\bY + \bE$ be the observations, including observation error $\bE$.
Denote the variance matrix of $\bY$ as $\Sigma$, and take the
observation error to be independent of $\bX$, with variance matrix
$T$.  If $V$ and $W$ are two variance matrices, then write $V \leq W$
exactly when $W - V$ is non-negative definite.  Then we have the
following result, which applies not just in the Gaussian case, but
also for more general second-order updating, such as the Bayes linear
approach described in \citet{gw07}.

\begin{theorem}\label{thm:bound}\samepage
  Let $\Sigma^* \ldef \Var(\bY \given \bZ)$ be the updated variance
  matrix of $\bY$.  If $\Sigma + T$ is non-singular, then $\Sigma^*
  \leq \Sigma$ and $\Sigma^* \leq T$.
\end{theorem}

One immediate use for this result is in code verification, but it is
also the basis of further results below, which are used for
visualisation.

\begin{proof}
  As $\Cov(\bY, \bZ) = \Sigma$ and $\Var(\bZ) = \Sigma + T$, the
  updated variance of $\bY$ satisfies
\begin{equation}
  \label{eq:bound1}
  \Sigma^* = \Sigma - \Sigma (\Sigma + T)^{-1} \Sigma 
\end{equation}
\citep[see, e.g.,][chapter~3]{mkb79}.  Hence $\Sigma^* \leq \Sigma$
because the second term on the righthand side of \eqref{eq:bound1} is
non-negative definite.  If we can show that
\begin{equation}
  \label{eq:bound2}
  \Sigma - \Sigma (\Sigma + T)^{-1} \Sigma = T - T(\Sigma + T)^{-1} T ,
\end{equation}
then \eqref{eq:bound1} and the same reasoning implies that $\Sigma^*
\leq T$, completing the proof.  Start with the identities
\begin{equation}
  \label{eq:bound3}
  \bs{0} =
  \left\lbrace 
  \begin{array}{c}
    \Sigma - \Sigma ( \Sigma + T)^{-1} (\Sigma + T) , \\
    T - T ( \Sigma + T)^{-1} (\Sigma + T) .
  \end{array} \right.
\end{equation}
Equating the two terms on the righthand side and rearranging gives
\begin{displaymath}
  \Sigma - \Sigma (\Sigma + T)^{-1} \Sigma - \Sigma (\Sigma + T)^{-1} T = T - T (\Sigma + T)^{-1} T - T (\Sigma + T)^{-1} \Sigma .
\end{displaymath}
But the final terms on each side of this expression are equal, because
they are symmetric, and \eqref{eq:bound2} is proved.  [See
\citet[section~1.2]{piziak07} for more general resuts of this type.]
\end{proof}

It is important that this result holds for singular $\Sigma$, provided
that ${\Sigma + T}$ is non-singular.  This is because we may well have
replications in the observations; e.g.\ the same component of $\bX$
observed several times.  This would be represented as duplicate rows
in $A$.  Alongside replications, we may well have more observations
than components of the state vector, e.g.\ if we have multiple
instruments with overlapping footprints.  This would be represented by
an $A$ with more rows than columns.  In both of these cases
\begin{displaymath}
  \Sigma = A \Var(\bX) A\trans
\end{displaymath}
would be singular (non-negative definite but not positive definite).
No matter what the form of $A$, a non-singular $T$ is sufficient for
$\Sigma + T$ to be non-singular (positive definite).  Thus
Theorem~\ref{thm:bound} always holds if there is measurement error.

\subsection{Local and global updating}
\label{sec:local}

A further useful result concerns the relationship between the
\emph{joint} update $\Var(Y_i \given \bZ)$ and the \emph{local} update
$\Var(Y_i \given Z_i)$.  This is a special case of the following more
general result about nested updates.

\begin{theorem}\label{thm:local}\samepage
  Let $B$ and $B'$ be two subsets of the indices of $\bY$, with $B
  \supset B'$.  Then
  \begin{equation}
    \label{eq:local1}
    \Var( \bY_{\!\!B'} \given \bZ_B ) \leq \Var( \bY_{\!\!B'} \given \bZ_{B'} ) \leq T_{B',B'} .
  \end{equation}
\end{theorem}

\begin{proof}
  It is a standard result in second-order updating that $B \supset B'$
  implies that
  \begin{displaymath}
    \Var (\bY \given \bZ_B) \leq \Var ( \bY \given \bZ_{B'} ) ,
  \end{displaymath}
  \citep[see, e.g.,][section~5.2]{gw07}.  This implies the first
  inequality in \eqref{eq:local1}, because $\bY_{\!\!B'}$ is a set of
  linear combinations of $\bY$.  Theorem~\ref{thm:bound} implies that
  $\Var(\bY_{\!\!B'} \given \bZ_{B'}) \leq T_{B',B'}$ for any $B'$,
  which is the second inequality in \eqref{eq:local1}.
\end{proof}

The interesting feature of this result is the introduction of $T$ in
place of $\Sigma$ in the second inequality, courtesy of
Theorem~\ref{thm:bound}.  Often the measurement error variance $T$ is
much smaller than the initial variance $\Sigma$, and so
Theorem~\ref{thm:local} represents a substantial lowering of the upper
bound on any updated variance.

The next result follows immediately from Theorem~\ref{thm:local},
taking $B$ to be the complete set of indices and $B' = \{i\}$:
\begin{equation}
  \label{eq:local2}
  \Var(Y_i \given \bZ) \leq \Var(Y_i \given Z_i) \leq T_{ii} \qquad \text{for each $i$.}
\end{equation}
This ordering of global, local, and observation error variances is
used our visualisation tool, presented in section~\ref{sec:medal}.  We
can verify the second inequality in \eqref{eq:local2} by direct
calculation:
\begin{equation}
  \label{eq:local3}
  \Var( Y_i \given Z_i ) 
  = \Sigma_{ii} - \frac{ \Sigma_{ii} \cdot \Sigma_{ii} }{ \Sigma_{ii} + T_{ii} }
  = \frac{ \Sigma_{ii} \cdot T_{ii} }{ \Sigma_{ii} + T_{ii} } \leq T_{ii} .
\end{equation}
This expression shows that there is a limit to how much relative
effect a local update can have.  Taking $T_{ii} \leq \Sigma_{ii}$, for
concreteness,
\begin{equation}
  \label{eq:inf}
  \inf_{ T_{ii} \leq \Sigma_{ii} } \frac{ \Var(Y_i \given Z_i) }{ T_{ii} } 
  = \inf_{ \kappa \leq 1 } \frac{ 1 }{ 1 + \kappa } 
  = \frac{ 1 }{ 1 + 1 } = \frac{ 1 }{ 2 }.
\end{equation}
In other words, information from $Z_i$ alone can push the updated
variance of $Y_i$ down to half of its upper bound, and this occurs
when $\Sigma_{ii} = T_{ii}$.  Eq.~\eqref{eq:local3} also shows that if
$T_{ii} \ll \Sigma_{ii}$ then $\Var(Y_i \given Z_i) \approx T_{ii}$.
In other words, the variance of the local update tends to the
observation error variance as the observation error variance becomes
small relative to the initial variance.

\subsection{Limiting behaviour}
\label{sec:limit}

The case where one of $\Sigma$ or $T$ is much larger than the other
occurs frequently in practice, and it is interesting to consider the
limiting case where, for concreteness, $T$ becomes vanishingly small
relative to $\Sigma$.  However, there is a difficulty with this case:
if $\Sigma$ is singular, then a `vanishingly small' $T$ will
ultimately conflict with the requirement that ${\Sigma + T}$ be
non-singular.  But, as explained in section~\ref{sec:bound}, it is
common for $\Sigma$ to be singular.  Therefore the following result
has additional conditions relative to Theorems~\ref{thm:bound} and
\ref{thm:local}, but it is powerful when these conditions hold.

\begin{theorem}\label{thm:limit}
  Let $\Sigma$ and $T$ both be non-singular, and define $r \ldef
  \norm{ T \Sigma^{-1} }$, where $\norm{ \cdot }$ is any induced
  $p$-norm.  If $r < 1$ then
  \begin{equation}
    \label{eq:limit1}
    \norm{ \Sigma^* - T } \leq 
    \frac{ \norm{ \Sigma^{-1} }\, \norm{T}^2 }{ 1 - r } .
  \end{equation}
\end{theorem}

\begin{proof}
  Start from \eqref{eq:bound1} and the top branch of \eqref{eq:bound3}
  to show that
  \begin{displaymath}
    \Sigma^* = \Sigma ( \Sigma + T)^{-1} T .
  \end{displaymath}
  Now under the conditions of the Theorem both $\Sigma$ and $T$ are
  non-singular, and this expression can be rearranged to show that
  \begin{displaymath}
    \Sigma^* = ( \Sigma^{-1} + T^{-1} )^{-1} 
  \end{displaymath}
  \citep[see also][section~2.3.3]{rue05}.  Then \eqref{eq:limit1}
  follows from a standard result about inverses and perturbations
  \citep[see, e.g.,][Theorem~2.3.4]{golub96}.
\end{proof}

In other words, if both $\Sigma$ and $T$ are non-singular then as $T$
becomes small relative to $\Sigma$, so the updated variance converges
to $T$.  However, it is important to appreciate that $T$ non-singular
is not, on its own, sufficient for this convergence.  This is seen in
the following counter-example with a singular $\Sigma$:
\begin{displaymath}
  \Var(\bX) = 10^6
  \begin{pmatrix}
    1.0 & 0.4 \\ 0.4 & 1.0 \\ 
  \end{pmatrix}, \quad
  A =
  \begin{pmatrix}
    1.0 & 1.0 \\ 1.0 & 0.0 \\ 0.0 & 1.0 \\ 
  \end{pmatrix}, \quad
  T =
  \begin{pmatrix}
    1.0 & 0.0 & 0.0 \\ 0.0 & 0.1 & 0.0 \\ 0.0 & 0.0 & 0.1 \\
  \end{pmatrix} ,
\end{displaymath}
for which, informally, $T \ll \Sigma = A \Var(\bX) A\trans$.  But
\begin{displaymath}
  \Sigma^* = 
  \begin{pmatrix}
    0.17 & 0.08 & 0.08 \\ 0.08 & 0.09 & -0.01 \\ 0.08 & -0.01 & 0.09 \\    
  \end{pmatrix} .
\end{displaymath}
It can be checked that $\Sigma^* \leq T$, as required by
Theorem~\ref{thm:bound}, but clearly $\Sigma^* \centernot\approx T$.
This combination of a singular $\Sigma$ with both `large footprint'
imprecise observations and `small footprint' precise observations
occurs in our illustration in section~\ref{sec:AIS}.

\paragraph{Principle of stable inference.}

Theorems \ref{thm:bound} and \ref{thm:limit} provide a multivariate
generalisation of L.J. Savage's principle of stable inference,
mentioned in the Introduction.  We restate it here.  The updated
variance in this statement is a general second-order update, but
applies in particular when $\bY$ and $\bE$ are both Gaussian.

\begin{theorem}[Principle of stable inference, multivariate]\samepage
  Let $\bY$ be a vector of random quantities and $\bZ$ be a vector of
  noisy observations on the components of $\bY$.  Define the
  measurement error as $\bE \ldef \bZ - \bY$.  Let $\bE \Perp \bY$ and
  $\Var(\bE)$ be non-singular.  Then $\Var(\bY \given \bZ) \leq
  \Var(\bE)$.  Furthermore, if $\Var(\bY)$ is non-singular and
  $\Var(\bE) \ll \Var(\bY)$ then $\Var(\bY \given \bZ) \approx
  \Var(\bE)$.
\end{theorem}

This principle underlies the common `plug-in' approximation
\begin{displaymath}
  \text{truth} \given \text{measurement} = \text{measurement} + \text{measurement error} .
\end{displaymath}
Our results indicate that the critical modelling judgement under which
this approximation provides a conservative or approximate assessment
of uncertainty about the true values $\bY$ is that the (additive)
measurement error $\bE$ is probabilistically independent of $\bY$.
The other conditions seem much less demanding in practice.

\subsection{Visualisation: the `medal plot'}
\label{sec:medal}

We would like to visualise various features of the variance update,
particularly for those components of $\bY$ which correspond to spatial
locations.  These features include, for a specified $Y_i$: what the
upper bound is, and where it comes from; what the updated variance is;
and what contribution is made by observations other than $Z_i$.

These considerations suggest the following `medal plot'.  Each $Y_i$
is represented by a medal of three concentric disks of decreasing
radius:
\begin{enumerate}\samepage
\item A red/blue disk representing the upper bound on the updated
  variance of $Y_i$, either \emph{red} where the prior provides the
  upper bound, or \emph{blue} where the observation error provides the
  upper bound.
\item A \emph{white} disk representing the updated variance using
  $Z_i$ alone (local update).
\item A \emph{gold} disk representing the updated variance using all
  observations (joint update).
\end{enumerate}
In all cases, the radius of the disk is proportional to the standard
deviation.

The medals can be scaled so that when displayed on a map they do not
overlap by more than is necessary to preserve the systematic spatial
patterns.  When there is an overlap, it is more effective to plot all
of the red/blue disks first, and then to overplot with the white
disks, and then with the gold disks.  In some application, including
our illustration below, it is more effective to use a semi-transparent
light-blue than white, so that underlying map features are preserved.

We recapitulate the properties of each medal, based on the results of
the previous subsections.  First, the three disks must be nested, with
gold inside white inside red/blue; see \eqref{eq:local2}.  Hence, each
medal appears as a gold central disk, a white annulus, and a red/blue
rim.

Second, the outer (red/blue) rim cannot be thicker than $1 -
1/\sqrt{2} \approx 0.3$ of the total radius of the medal, because the
update from the local observation alone cannot reduce the updated
variance to less than one half of its upper bound; i.e.\ the white
disk cannot have a radius less than $1/\sqrt{2} \approx 0.7$ of the
total radius; see \eqref{eq:inf}.  Moreover, a very thin rim indicates
a large discrepancy between the initial variance and the measurement
error variance.  Thus a very thin blue rim indicates that the initial
variance is far larger than the observation error variance.

Third, if both $\Sigma$ and $T$ are non-singular, then the combined
thickness of the (red/blue) rim and the white annulus is compressed
towards zero when one of the two initial variances ($\Sigma$ or $T$)
dominates the other, because in this case the updated variance (gold
disk) is almost the same as the upper bound (red/blue disk); see
section~\ref{sec:limit}.

For a given medal at location $i$, we might be particularly interested
in the thickness of the white annulus.  This thickness shows us how
much of the update of $Y_i$ is coming from observations other than
$Z_i$, with a thick annulus showing that other observations are making
a large contribution (i.e.\ driving the updated variance well below
what is achieved by $Z_i$ alone).  When we compare the medals across
the domain of the observations we can see at a glance how the spatial
scale of the update varies, by comparing the widths of the white
annuli.

\subsection{Model parameters}
\label{sec:parameters}

In many cases, the second-order structure of $(\bX, \bY, \bZ)$ is
expressed conditionally on uncertain parameters $\theta$; these might
be spatial correlation lengths in $\bX$, for example.  But there is no
data-free bound in these cases, because
\begin{displaymath}
  \Var(\bY \given \bZ) = \E \{ \Var(\bY \given \bZ, \theta) \given \bZ\} + \Var \{ \E( \bY \given \bZ, \theta) \given \bZ \} ,
\end{displaymath}
and both terms will depend on the value of $\bZ$, which affects the
update of $\theta$, and also the value of $\E(\bY \given \bZ,
\theta)$.  In this situation, where diagnostics and visualisation are
the goal, it is necessary to plug-in a point value for $\theta$.  Such
a plug-in suffices to trap coding errors, because the lower bound must
be respected for every value of $\theta$.  If $\theta$ is
well-constrained by $\bZ$, then the plug-in visualisation based on the
posterior mean for $\theta$ will only differ a little from that with
$\theta$ integrated out.

\section{Computation}
\label{sec:computation}

A full-scale update of the expectation and variance of a large-scale
Gaussian random field with many observations can be prohibitively
expensive.  Here we provide details for the much cheaper option of
only updating the diagonal components of the variance matrix of the
observations, which is all that is required for the medal plot
visualisation we described in section~\ref{sec:medal}.

When the state vector $\bX$ is very large, standard calculations that
start with the specification of $\Xi \ldef \Var(\bX)$ are no longer
feasible, because of the cost of holding the elements of $\Xi$ in
memory while computing products such as $\Var(\bY) = A \Xi A\trans$.
In this case, one powerful strategy is to approximate the distribution
of $\bX$ using a Gauss Markov random field (GMRF), as described in
\cite{rue05} and \citet{rue07}.  A GMRF representation for $\bX$
allows a very efficient calculation for the elements of a medal plot,
as we now discuss.

In a GMRF, the variance parameter of $\bX$ is $Q \ldef \Xi^{-1}$.  A
high degree of conditional independence (e.g.\ arising from judgements
of smoothness, or from $\bX$ being the composition of several
independent processes) implies that $Q$ is sparse, with the location
of its zeros being known.  This in turn implies that the Choleski
factor $L$ is also sparse, where $L$ is lower-triangular and $L
L\trans = Q$.  There are efficient algorithms for computing the
non-zero elements of $L$ from $Q$; see, for example,
\citet[section~2.4]{rue05}.

The next stage towards finding $\Xi$ would be to compute this from
$L$: notionally $\Xi = L\mtrans L^{-1}$.  The efficient algorithm for
this stage is based on the Takahashi equation, formalised in the
result of \citet{erisman75}.  If $L_{ij} \neq 0$, then $\Xi_{ij}$ can
be computed directly from $L$ and values of $\Xi_{pq}$ for which $p
\geq i$, $q \geq j$, and $L_{pq} \neq 0$.  The $(i,j)$ elements of
$\Xi$ for which $L_{ij} \neq 0$ are termed the \emph{sparse subset} of
$\Xi$.  The other elements of $\Xi$ can also be computed, \emph{but
  only if needed}.  This final observation is crucial for efficiency
in our application, because if $A$ is sparse then many of the elements
of $\Xi$ that are not in the sparse subset are not needed for the
medal plot.

For our medal plot we require the three vectors
\begin{displaymath}
  \diag (A \Xi A\trans), \quad \diag (A \Xi^* A\trans), \quad \text{and } \diag T ,
\end{displaymath}
where $\Xi^* \ldef \Var(\bX \given \bZ)$.  In fact, we may not want
all components of these three vectors, if we only want to visualise a
subset of the observations.  In this case, we would simply drop rows
from $A$, and the corresponding rows and columns from the observation
error variance matrix $T$.

Consider the first vector, whose $i$th element is
\begin{displaymath}
  [\diag (A \Xi A\trans)]_i = \sum_k \sum_j A_{ij} \cdot \Xi_{jk} \cdot A_{ik} 
  = \sum_j \sum_k (A_{ij} \cdot A_{ik}) \cdot \Xi_{jk} .
\end{displaymath}
Hence $\Xi_{jk}$ influences the $i$th diagonal element when both
$A_{ij} \neq 0$ and $A_{ik} \neq 0$.  Therefore $\Xi_{jk}$ does not
influence the vector of diagonal elements if, for every $i$, either
$A_{ij} = 0$ or $A_{ik} = 0$.  Now suppose that all of the elements of
$A$ are non-negative.  In this case
\begin{equation}
  \label{eq:zeros0}
  [A\trans A]_{jk} = 0 \iff \text{$A_{ij} = 0$ or $A_{ik} = 0$ for every $i$.}
\end{equation}
Therefore if $[A\trans A]_{jk} = 0$ then $\Xi_{jk}$ is not required in
order to compute the elements of $\diag (A \Xi A\trans)$.  Hence if
$A$ is sparse, such that many elements of $A\trans A$ are zero, then
only a few extra elements of $\Xi$ will be needed, beyond the sparse
subset.  And if $Q$ is sparse, then $L$ is sparse and the sparse
subset of $\Xi$ is small.

Exactly the same sequence of operations applies for the updated
variance $\Xi^*$, starting from the updated precision matrix
\begin{displaymath}
  Q^* \ldef Q + A\trans T^{-1} A
\end{displaymath}
\citep[section~2.3.3]{rue05}.  This updated precision matrix is sparse
if $Q$ and $A$ are sparse, and if $T$ is diagonal, or block diagonal
with a small bandwidth.  But there is a much more powerful result if
$A$ is non-negative and $T$ is diagonal, as we now show.

First, introduce a new operator,
\begin{displaymath}
  \zeros(A)_{ij} \ldef
  \begin{cases}
    0 & A_{ij} = 0 \\ 1 & \text{otherwise.}
  \end{cases}
\end{displaymath}
Note that $A$ non-negative implies that
\begin{equation}
  \label{eq:zeros==}
  [\zeros(A\trans A)]_{jk} = 0 \iff [A\trans A]_{jk} = 0 .
\end{equation}
Now let $T$ be diagonal, in which case
\begin{displaymath}
  \zeros(Q^*) = \zeros(Q + A\trans T^{-1} A) \geqquery \zeros(A\trans T^{-1} A) = \zeros(A\trans A) .
\end{displaymath}
The second inequality, marked as `$\geqquery$', needs a qualification.
It is \emph{not} the case that $\zeros(M + N) \geq \zeros(N)$ for
arbitrary $M$ and $N$.  However, as in \citet{erisman75}, we will
assume that every $[Q + A\trans T^{-1} A)]_{ij}$ which must be
computed is treated as nonzero, even if its value is zero due to
numerical cancellation. The inequality follows if we adopt this
treatment.

To continue, note that $\zeros(L^*) \geq \zeros(Q^*)$ in the lower
triangle of $Q^*$, where $L^*$ is the lower triangular Choleski factor
of $Q^*$; see \citet[Corollary~2.2, section~2.4]{rue05}.  Thus we
have, for $j \geq k$,
\begin{displaymath}
  \zeros(L^*)_{jk} \geq \zeros(Q^*)_{jk} \geq \zeros(A\trans A)_{jk} .
\end{displaymath}
Hence if ${\zeros(L^*)_{jk} = 0}$ then ${\zeros(A\trans A)_{jk} = 0}$,
which implies that ${[A\trans A]_{jk} = 0}$ from \eqref{eq:zeros==},
which implies that $\Xi^*_{jk}$ is not required in order to compute
the elements of $\diag (A \Xi^* A\trans)$, from \eqref{eq:zeros0}.  In
summary, we have proved the following result.

\begin{theorem}\samepage
  If $A$ is non-negative and $T$ is diagonal, then the elements of
  $\Xi^*$ required in order to compute $\diag(A \Xi^* A\trans)$ are a
  subset of the sparse subset of $\Xi^*$.
\end{theorem}

This result implies that the diagonal elements of $A \Xi^* A\trans$
are available `for free' once we have computed the sparse subset of
$\Xi^*$.  The conditions on $A$ and $T$ are both very natural for
large-scale spatial and spatial-temporal modelling.  Indeed, they were
a feature of our illustration well before we derived this result.

\section{Illustration}
\label{sec:AIS}

Our illustration is part of a mass-balance estimate for the Antarctic
Ice Sheet (AIS), which is the world's largest freshwater reservoir.
Here we provide the briefest outline of our inference, which we
describe in detail elsewhere \citep{zammit14,zammit14a}.

In order to determine the AIS contribution to sea-level change, we
need to decompose the change in height of the AIS over a fixed time
period into the sum of four main processes: change in the height of
the underlying rock, effect of ice dynamics, firn compaction
(densification of past years' snow), and the net effect of surface
processes (precipitation, run-off, melt, and refreeze). Then to
quantify the contribution to sea-level change, we sum the changes in
height of ice, firn, and surface processes inside the grounding line
(see the caption to Figure~\ref{fig:GRACE}) over the AIS, and then map
those to mass changes using specified densities.

We have observations from three types of instrument.  First, a small
number of GPS receivers on rocky outcrops, which give accurate
observations for change in height of the underlying rock (at those
outcrops).  Second, satellite altimetry, which gives observations of
height change (i.e.\ summing the four processes) along specified
transects.  Third, satellite gravimetry (Gravity Recovery and Climate
Experiment, or GRACE), which provides measures of mass change, and
therefore sees a linear combination of change in the height of the
underlying rock, the ice dynamics and surface processes (firn
compaction changes height but not mass).  These three instruments have
very different spatial footprints, with GPS being a point observation,
altimetry having a footprint of ${\sim\!1 \unit{km}}$ (treated as a
point observation), and gravimetry having a footprint of ${\sim\!400
  \unit{km}}$.

This is an inherently statistical problem because: (i)~we have three
instruments for four fields; (ii)~there are substantial observation
errors, (iii)~the footprints of the instruments are of such different
sizes, (iv)~the observations do not cover the whole of the AIS, and
(v)~uncertainty assessment is a crucial output for impact studies
related to sea-level rise.  The problem becomes soluble once we
incorporate prior information about the processes, notably their
variabilities and their characteristic length scales, both of which
can vary spatially.  As well as the four fields, our unknowns include
statistical parameters for the processes and in the observation
equation.

For this illustration we used finite element basis functions to model
each of the four processes \citep[see, e.g.,][]{lindgren11}, with
variable resolution to account for greater heterogeneity near the
coastline.  We used a blocked Gibbs sampler to update the processes
conditional on the statistical parameters, and to update the
statistical parameters conditional on the processes.  Then we plugged
in the maximum \textit{a posteriori} estimate of the statistical
parameters (which were well-constrained), and redid the update of the
fields, to draw the medal plots.  We illustrate with a medal plot for
the gravimetry observations for 2006, shown in Figure~\ref{fig:GRACE}.
Recollect that the medals show the update from all observations; e.g.\
the gravimetry linear combinations are updated not just by the
gravimetry observations, but also by GPS and altimetry.

\begin{figure}
  \includegraphics[width=\textwidth]{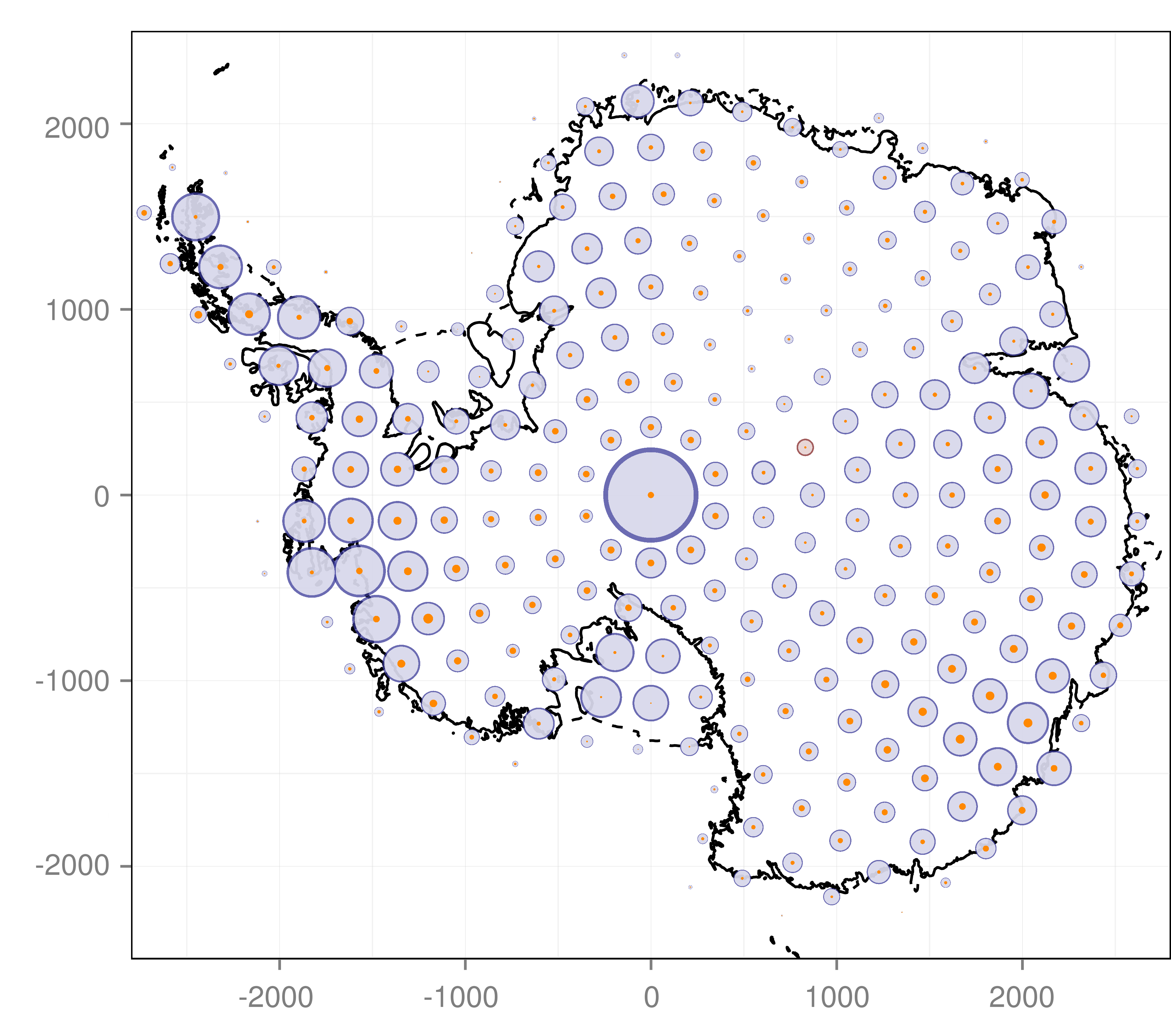}
  \caption{Medal plot for GRACE observations over Antarctica, with
    distances in kilometres.  The solid line is the grounding line
    (where the ice floats), and the dashed line is the coastline (the
    limit of ice extent).  See section~\ref{sec:AIS} for details of
    the application and observations, and section~\ref{sec:medal} for
    the interpretation of the medals.  We have used a semi-transparent
    grey instead of white for the annulus.}
  \label{fig:GRACE}
\end{figure}

We highlight some of the wealth of information in this figure.  First,
almost all of the medals have blue rims, showing that the upper bound
on the updated variance comes from the observation error variance, not
the initial variance.  There is one exception, which is at about
$(+800 \unit{km}, +250 \unit{km})$.  Here our model for mass trends
implies a small initial expectation and variance, because the ice
velocity and expected accumulation is so small.  Moreover, the rims
are all very thin, indicating (in the case of the blue rims) that the
initial variance is much larger than the observation error variance.
(In fact, in our plotting we expand the rims slightly, where they
would otherwise be hardly visible.)

Second, there are clear spatial patterns in the observation error
uncertainties.  These uncertainties are provided along with the GRACE
observations, although we have had to infer a covariance structure.
These patterns in the uncertainties are related to physical features
that induce variations in the GRACE observations for successive
overpasses of the satellites.  For example, uncertainty outside the
grounding line tends to be small, because the (floating) ice is in
hydrostatic equilibrium.  Around the grounding line the uncertainty
tends to be relatively large because of variations in precipitation
and ice velocity.  The large medal at the South Pole simply reflects a
GRACE observation with a larger spatial footprint, owing to how we
have regridded the observations.

Third, the globally updated uncertainties are much smaller than the
locally updated uncertainties, as indicated by a thick grey annulus.
Therefore much of the reduction in uncertainty at each location is
coming from other observations.  Some of this will be from other GRACE
observations, because GRACE sees height changes in the underlying
rock, which has a very long correlation length (i.e.\ is spatially
very stiff).  But some of it might also come from the altimetry
observations.  While we could do separate medal plots to quantify each
contribution, in practice we do not have to.  Altimetry observations
are dominated by surface processes with short correlation lengths.
But altimetry satellites cannot overfly the South Pole, and hence
there is no altimetry contribution to the South Pole medal.  Since the
updated variance is about the same at the South Pole as the other
medals, we conclude that it is other GRACE observations that dominate
the global update shown in each GRACE medal.

These rationalisations of the GRACE medal plot increase our confidence
in our statistical modelling and also in our computation.  Experienced
modellers will appreciate that earlier medal plots of these and the
other observations in this application presented apparent anomalies
which we were unable to rationalise or verify through testing.  We
traced these back to modelling or computing choices that we
subsequently revisited.

\section*{Acknowledgements}

We would like to thank Finn~Lindgren and Botond~Cseke for several helpful
discussions on the computations.

\singlespacing\small
\bibliography{statistics,dynamical,climate,additional,palaeo}

\end{document}